%% file: main.tex
\documentclass[11pt]{article}

\usepackage[margin=1in]{geometry}
\usepackage{amsmath, amssymb, amsthm}
\usepackage{booktabs}
\usepackage{longtable}
\usepackage{hyperref}
\usepackage{graphicx}
\usepackage{microtype}
\usepackage{listings}
\usepackage{tikz}
\usepackage{url}
\usepackage{enumitem}
\usepackage{placeins}
\usepackage{float}
\usepackage{array}

\setlist[itemize]{noitemsep, topsep=2pt}
\hypersetup{
  colorlinks=true,
  linkcolor=blue,
  citecolor=blue,
  urlcolor=blue
}

\newtheorem{theorem}{Theorem}
\newtheorem{lemma}{Lemma}
\newtheorem{definition}{Definition}

\title{Symfrog-512: High-Capacity Sponge-Based AEAD Cipher (1024-bit State)}
\author{
Victor Duarte Melo\\
Independent Researcher\\
Erd\H{o}s Number 4... 
}
\date{}

\begin{document}
\maketitle

\begin{abstract}
SymFrog-512 is an authenticated encryption with associated data (AEAD) scheme built from a custom 1024-bit permutation, denoted \textsf{P1024-v2}, and a duplex sponge mode with rate $r = 512$ bits and capacity $c = 512$ bits. The design combines a large capacity with a 256-bit authentication tag and a 256-bit nonce, targeting conservative generic bounds in the ideal permutation model. The reference implementation supports either a raw 1024-bit key or a passphrase-derived key via Argon2id (libsodium), and it defines an encrypted file format with a dedicated keyed header authentication tag for early reject and robust parsing. This paper provides a complete specification aligned with the current reference code, a security analysis in the ideal permutation model using standard sponge and duplex arguments, and empirical measurements (benchmarks and diffusion checks) with regression artifacts and checksums.

Reference implementation: \url{https://github.com/victormeloasm/symfrog512}. Reported benchmark run: 435.1 ns per \textsf{P1024-v2} permutation (200,000 iterations) and 131.7 MiB/s for the AEAD core encrypt path on a 64 MiB buffer (excluding file I/O and key derivation).
\end{abstract}

\section{Introduction}

Permutation-based cryptography has a long history, but the modern sponge framework gave it a unified and implementer-friendly structure. In this paradigm, one designs a fixed-width permutation and then builds multiple primitives by defining how messages are absorbed into, and squeezed out of, an internal state. Keccak and SHA-3 popularized this split between a public permutation and a mode of use \cite{fips202, bertoni_sponge}. Duplexing generalizes the sponge to interactive settings, enabling single-pass authenticated encryption and streaming constructions \cite{bertoni_duplex}.

SymFrog-512 is an experimental AEAD design in this spirit. It defines a 1024-bit state permutation, \textsf{P1024-v2}, and uses a duplex sponge mode with a 512-bit rate and a 512-bit capacity. The design aims for conservative generic bounds (large capacity, long tag) while remaining implementable with simple 64-bit operations and a clear file format.

This document has three goals:
(i) provide a precise specification that matches the current public reference implementation,
(ii) provide a transparent security discussion grounded in the ideal permutation model and standard sponge bounds,
(iii) provide portable artifacts: benchmark summaries, diffusion sanity checks, and regression files with SHA-256 checksums.

\paragraph{Non-goals.} This paper does not claim that \textsf{P1024-v2} has been proven secure against all dedicated cryptanalysis. The security statements are reductions in the ideal permutation model, which is the standard baseline for sponge and duplex modes. Public cryptanalysis is encouraged.

\section{Related Work}

The sponge function and the duplex construction provide generic security bounds and design guidance for permutation-based schemes \cite{bertoni_sponge, bertoni_duplex}. NIST standard FIPS 202 specifies SHA-3 and SHAKE based on the Keccak permutation \cite{fips202}. Many modern AEADs either use block ciphers (AES-GCM), stream ciphers and MACs (ChaCha20-Poly1305), or dedicated permutations (Ascon, Xoodyak, and others). SymFrog-512 follows the permutation-based path and uses a duplex interface with explicit domain separation bytes.

Argon2id is a password hashing and key derivation function designed to be memory-hard and to limit massively parallel attacks on low-entropy passwords \cite{argon2}. The SymFrog-512 reference implementation optionally uses Argon2id via libsodium \cite{libsodium} for passphrase-derived keys.

\section{Parameters, Notation, and Conventions}

\subsection{Bit and byte conventions}

All integers are unsigned. The permutation state is a vector of 16 little-endian 64-bit words. For a byte string $X$, $\lvert X\rvert$ is its length in bytes. Concatenation is written $X \,\|\, Y$. All arithmetic on 64-bit words is modulo $2^{64}$. Bit shifts and rotations are on 64-bit words, with $\textsf{ROTL}_k(x)$ denoting rotation left by $k$ bits.

\subsection{Fixed parameters}

The SymFrog-512 instantiation uses:
\begin{align*}
\text{State size} &= 1024 \text{ bits} = 128 \text{ bytes} \\
r &= 512 \text{ bits} = 64 \text{ bytes} \\
c &= 512 \text{ bits} = 64 \text{ bytes} \\
\text{Tag length } t &= 256 \text{ bits} = 32 \text{ bytes} \\
\text{Nonce length} &= 256 \text{ bits} = 32 \text{ bytes} \\
\text{Key length} &= 1024 \text{ bits} = 128 \text{ bytes} \\
\text{Rounds} &= 24
\end{align*}

The state is partitioned as $S = (S_R, S_C)$ where $S_R$ is the first 64 bytes (rate) and $S_C$ is the last 64 bytes (capacity).

\subsection{Padding on the rate}

For a final partial block of length $\ell$ bytes (with $0 \le \ell < 64$), SymFrog XORs the $\ell$ bytes into the rate, then XORs $0x80$ at byte position $\ell$, and XORs $0x01$ into the last rate byte (byte position 63). This is the common $10^*1$ padding pattern used in sponge designs.

\section{The \textsf{P1024-v2} Permutation}

\subsection{State layout}

The permutation state is $S[0],\ldots,S[15]$, 16 words of 64 bits. Words $S[0..7]$ correspond to the rate and $S[8..15]$ correspond to the capacity. Round constants are injected into the capacity words.

\subsection{Round constants}

Round constants are deterministically derived via SHAKE256:
\[
RC[r] = \textsf{SHAKE256}(\textsf{``SymFrog-rc-v1''} \,\|\, \textsf{LE32}(r))[0..63]
\]
The 64 output bytes are parsed as 8 little-endian 64-bit words and XORed into $S[8]$ through $S[15]$ at the start of round $r$.

\subsection{Round function}

A single round consists of the following layers, applied in order:
\begin{enumerate}[leftmargin=2em]
\item \textbf{AddRoundConstants:} XOR the 8-word constant $RC[r]$ into $S[8..15]$.
\item \textbf{Mixer:} for $i = 0$ to $7$, set $S[i] \gets S[i] \oplus S[i+8]$.
\item \textbf{Chi nonlinearity:} apply a 4-word chi to each group $(S[4g],S[4g+1],S[4g+2],S[4g+3])$.
\item \textbf{Kick layer:} two phases of 64-bit multiplicative mixing using an oddness mask.
\item \textbf{Rotate and shuffle:} rotate each word and permute word positions via a fixed P-box.
\end{enumerate}

\paragraph{Chi layer.}
For each group $(x_0,x_1,x_2,x_3)$:
\begin{align*}
x_0 &\gets x_0 \oplus (\neg x_1 \wedge x_2)\\
x_1 &\gets x_1 \oplus (\neg x_2 \wedge x_3)\\
x_2 &\gets x_2 \oplus (\neg x_3 \wedge x_0)\\
x_3 &\gets x_3 \oplus (\neg x_0 \wedge x_1)
\end{align*}

\paragraph{Kick layer.}
Let $C = \texttt{0x9E3779B97F4A7C15}$. Phase A updates odd indices from even indices: for $i \in \{0,2,\ldots,14\}$, define $m = S[i] \,|\, 1$ and set $S[i+1] \gets S[i+1] \oplus (S[i]\cdot m)$. Phase B updates even indices from odd indices: for $i \in \{1,3,\ldots,15\}$, define $m = S[i] \,|\, 1$, compute $k = S[i]\cdot (m \oplus C)$, and set $S[(i+1)\bmod 16] \gets S[(i+1)\bmod 16] \oplus \textsf{ROTL}_{23}(k)$.

\paragraph{Rotate and shuffle.}
Each word is rotated left by 19 bits if its index is even, and by 61 bits if its index is odd. Then apply the 16-word P-box:
\[
\pi = [0,13,10,7,4,1,14,11,8,5,2,15,12,9,6,3].
\]
After shuffling, the new state is $(S[\pi(0)],\ldots,S[\pi(15)])$.

\subsection{Rationale and expected properties}

The round function combines Boolean nonlinearity (chi), multiplication-based mixing (kick), and inter-word diffusion (mixer plus shuffle). The chi step is a classic way to introduce nonlinearity using only bitwise operations, and it appears in Keccak as well \cite{fips202}. The kick layer injects data-dependent carries through multiplication, which tends to raise algebraic complexity. The rotate and shuffle layer spreads dependencies across words.

These design elements are motivated by engineering intuition and by patterns used in other permutation designs, but they do not constitute a full cryptanalytic proof. The AEAD security discussion below therefore separates \emph{mode security} (ideal permutation model bounds) from \emph{permutation security} (open to public analysis).

\section{SymFrog-512 AEAD Mode}

\subsection{Domain separation}

SymFrog separates phases by XORing a one-byte domain constant into the least significant byte of $S[15]$ before each permutation call:
\[
DS_{AD} = \texttt{0xA0},\quad DS_{CT} = \texttt{0xC0},\quad DS_{TAG} = \texttt{0xF0}.
\]
The header authentication mechanism uses:
\[
DS_{HDR} = \texttt{0xB0},\quad DS_{HDRTAG} = \texttt{0xB1}.
\]

\subsection{Initialization}

Given a 1024-bit key $K$ and a 256-bit nonce $N$, initialization is:
\begin{enumerate}[leftmargin=2em]
\item Parse $K$ as 16 little-endian 64-bit words and set $S[i] \gets K_i$ for $i=0$ to $15$.
\item XOR the nonce into the last 4 words:
$S[12] \gets S[12]\oplus N_0$, $S[13] \gets S[13]\oplus N_1$, $S[14] \gets S[14]\oplus N_2$, $S[15] \gets S[15]\oplus N_3$.
\item XOR a fixed ASCII identifier into $S[8]$ through $S[11]$ as in the reference code (\texttt{SYMFROG-512-AEAD-v1} and a version word).
\item Apply the permutation once: $S \gets \textsf{P1024-v2}(S)$.
\end{enumerate}

\subsection{Associated data absorption}

Associated data $AD$ is absorbed in 64-byte blocks. For each full block $A_j$:
\[
S_R \gets S_R \oplus A_j,\quad S[15] \gets S[15] \oplus DS_{AD},\quad S \gets \textsf{P}(S).
\]
For the final partial block, apply the $10^*1$ padding on the rate and perform one permutation call with the same $DS_{AD}$.

Even when $AD$ is empty, the reference implementation performs one padded AD absorption step, which guarantees a consistent transcript separation between initialization and ciphertext processing.

\subsection{Output transform}
\label{sec:out}
Symfrog-512 does not expose raw state lanes directly. Instead, each duplex squeeze step extracts a 512-bit block by mixing rate and capacity words and applying a non-linear 64-bit finalizer on every lane.

Let $S=(S_0,\ldots,S_{15}) \in (\mathbb{Z}_{2^{64}})^{16}$. For $i\in\{0,\ldots,7\}$ define
\[
X_i = S_i \oplus \mathrm{ROTL}(S_{8+i},17) \oplus \mathrm{ROTL}(S_{8+((i+3)\bmod 8)},41) \oplus \gamma\cdot(i+1),
\]
where $\gamma=\texttt{0x9E3779B97F4A7C15}$ and all operations are modulo $2^{64}$. The extracted 64-bit word is
\[
\mathrm{Out}(S)[i] = \mathrm{SM}(X_i),
\]
where $\mathrm{SM}$ is the SplitMix64 finalizer (all arithmetic modulo $2^{64}$):
\[
\begin{aligned}
x &\leftarrow x \oplus (x \gg 30), & x &\leftarrow x\cdot \alpha,\\
x &\leftarrow x \oplus (x \gg 27), & x &\leftarrow x\cdot \beta,\\
x &\leftarrow x \oplus (x \gg 31), &&
\end{aligned}
\]
with $\alpha=\texttt{0xBF58476D1CE4E5B9}$ and $\beta=\texttt{0x94D049BB133111EB}$. The 512-bit output block is the concatenation of the eight words $\mathrm{Out}(S)[0..7]$ serialized in little-endian order.

The extraction uses both halves of the state, which improves diffusion from the capacity and breaks simple ``read-the-rate'' patterns often used in sponge constructions.
\subsection{Encryption and decryption}

Let the plaintext be $P$. Encryption proceeds in 64-byte blocks. For each full block:
\begin{enumerate}[leftmargin=2em]
\item Compute keystream $Z \gets \mathsf{Out}(S)$.
\item Output ciphertext $C \gets P \oplus Z$.
\item Absorb ciphertext into the rate: $S_R \gets S_R \oplus C$.
\item Domain separate and permute: $S[15] \gets S[15] \oplus DS_{CT}$, then $S \gets \textsf{P}(S)$.
\end{enumerate}
For the final partial block, ciphertext is produced by XOR with the corresponding keystream prefix; then the ciphertext tail is absorbed into the state with $10^*1$ padding on the rate, followed by one permutation call with $DS_{CT}$.

Decryption is identical except that plaintext is computed as $P \gets C \oplus \mathsf{Out}(S)$ while the state always absorbs ciphertext.

\subsection{Final tag}

After processing ciphertext, SymFrog finalizes by:
\[
S[15] \gets S[15]\oplus DS_{TAG},\quad S \gets \textsf{P}(S),\quad T \gets \mathsf{Out}(S)[0..31].
\]
The 32-byte tag $T$ is appended to the ciphertext stream.

\section{Encrypted File Format and Header Authentication}

\subsection{File layout}

A SymFrog encrypted file begins with a 152-byte header, followed by ciphertext bytes, followed by a 32-byte final tag. Total file size is \texttt{HEADER\_BYTES} + $\ell$ + \texttt{TAG\_BYTES} bytes, where $\ell$ is the plaintext length.

\begin{center}
\small
\begin{tabular}{l l l}
\toprule
Field & Size & Description \\
\midrule
magic & 8 & ASCII \texttt{SYMFROG1} \\
version & 4 & currently 1 \\
flags & 4 & bit 0 indicates Argon2id-derived key \\
salt & 32 & Argon2id salt (if derived key) \\
nonce & 32 & nonce used by AEAD \\
ct\_len & 8 & ciphertext length in bytes \\
reserved & 32 & reserved for future use \\
header\_tag & 32 & keyed header authentication tag \\
\bottomrule
\end{tabular}
\end{center}

\subsection{Key derivation}

If \texttt{FLAG\_KEY\_DERIVED} is set, the implementation derives a 1024-bit key from a passphrase using Argon2id v1.3 via libsodium \cite{argon2, libsodium}. The 32-byte \texttt{salt} field in the header is used as the Argon2id salt. If the flag is not set, a raw 1024-bit key is used directly.

\subsection{Header tag and early reject}

Before decrypting ciphertext, the implementation verifies \texttt{header\_tag}. This binds header fields plus external associated data to the key and nonce, and it prevents wasting time decrypting with the wrong key.

Let $H_0$ be the header with \texttt{header\_tag} bytes set to zero. The header tag is computed by running a keyed duplex transcript with domain bytes $DS_{HDR}$ and $DS_{HDRTAG}$:
\[
\texttt{header\_tag} = \mathsf{Tag32}\bigl(\mathsf{Init}(K,N);\ \mathsf{Absorb}(\textsf{``SYMFROG-HDRTAG-v1''}\,\|\,H_0\,\|\,AD, DS_{HDR});\ \mathsf{Finalize}(DS_{HDRTAG})\bigr).
\]

\subsection{Robustness against truncation}

For decryption, the implementation determines ciphertext length from the actual file size and checks that the last 32 bytes exist for the tag. Tag comparisons are constant-time. Output is written safely to avoid corrupting inputs.

\section{Design Rationale}

\subsection{Capacity, rate, and conservative generic bounds}

In sponge and duplex analyses, the capacity $c$ controls collision-based distinguishing bounds. A common generic term is about $q^2 / 2^c$ where $q$ counts permutation calls. Choosing $c=512$ makes this negligible for realistic file encryption workloads. The rate $r=512$ processes 64 bytes per permutation call, yielding a direct throughput to permutation-call relation.

\subsection{Long tag and long nonce}

The tag length $t=256$ yields a direct guessing probability $2^{-256}$ per forgery attempt. The 256-bit nonce is large enough to make random collisions across large collections of files extremely unlikely.

\subsection{Ciphertext absorption}

Absorbing ciphertext into the state enables symmetric transcripts for encryption and decryption and is consistent with standard duplex AEAD patterns.

\subsection{Output transform}
The duplex mode uses the extraction function $\mathrm{Out}(S)$ defined in Section~\ref{sec:out} to generate each 512-bit keystream block.
\section{Security Analysis in the Ideal Permutation Model}

This section states standard security bounds under the assumption that \textsf{P1024-v2} behaves like a uniformly random permutation on 1024-bit strings. This is the common baseline for sponge analyses \cite{bertoni_sponge, bertoni_duplex}.

\subsection{Security notions}

We consider standard AEAD goals:
\begin{itemize}
\item \textbf{Confidentiality (IND-CPA):} ciphertexts should not reveal which of two equal-length plaintexts was encrypted.
\item \textbf{Authenticity (INT-CTXT):} it should be infeasible to produce a new valid ciphertext and tag for some associated data, except with negligible probability.
\end{itemize}

\subsection{Capacity collisions}

\begin{definition}[Capacity collision event]
Let $S^{(i)}=(R^{(i)},C^{(i)})$ denote the internal state just before the $i$-th permutation call, where $C^{(i)}$ is the $c$-bit capacity portion. A capacity collision occurs if there exist $i<j$ with $C^{(i)}=C^{(j)}$.
\end{definition}

\begin{lemma}[Birthday bound]
If the capacity values are uniform and independent until the first collision, then
\[
\Pr[\textsf{CapCollision}] \le \frac{q(q-1)}{2^{c+1}},
\]
where $q$ is the number of permutation calls.
\end{lemma}

\subsection{Confidentiality bound}

\begin{theorem}[IND-CPA bound]
Let $\Pi$ be SymFrog-512 instantiated with an ideal permutation. For any adversary making at most $q$ permutation calls in total,
\[
\mathrm{Adv}^{\mathrm{ind\mbox{-}cpa}}_{\Pi} \le \frac{q(q-1)}{2^{c+1}}.
\]
\end{theorem}

\begin{proof}[Proof sketch]
The standard game-hopping argument replaces the real transcript with a simulated transcript that returns uniformly random keystream blocks, except when a capacity collision occurs. If no collision occurs, each keystream block is information-theoretically independent of the adversary's view, and ciphertext is a one-time pad on the chosen plaintext. The statistical distance between the real and ideal games is bounded by the collision probability.
\end{proof}

\subsection{Authenticity bound}

\begin{theorem}[INT-CTXT bound]
Let $\Pi$ be SymFrog-512 instantiated with an ideal permutation. For any adversary making at most $q$ permutation calls and then outputting one forgery attempt,
\[
\mathrm{Adv}^{\mathrm{int\mbox{-}ctxt}}_{\Pi} \le \frac{q(q-1)}{2^{c+1}} + 2^{-t},
\]
where $t=256$ is the tag length.
\end{theorem}

\begin{proof}[Proof sketch]
Condition on the event that no capacity collision occurs, which is bounded by the birthday term. Without collisions, the final tag is a deterministic function of a hidden internal state that remains information-theoretically unpredictable from the adversary's view, so the remaining success probability is at most the probability of guessing the $t$-bit tag.
\end{proof}

\subsection{Header tag}

The header tag uses separate domain bytes and processes a disjoint transcript (header fields plus associated data). In the ideal permutation model, forging the header tag without collisions requires guessing 256 bits, yielding the same $2^{-t}$ term.

\subsection{Nonce reuse}

The construction is not intended to be nonce-misuse resistant. Reusing a nonce with the same key can reveal XOR relations between plaintext blocks, as with other stream-like AEADs. The file format stores the nonce explicitly to reduce operational errors.

\section{Empirical Evaluation}

\subsection{Benchmarks}

Figure \ref{fig:benchperm} shows the permutation latency from the provided benchmark run. Figure \ref{fig:benchaead} shows the measured throughput of the AEAD core encrypt path on a 64 MiB buffer, excluding file I/O and key derivation.

\begin{figure}[t]
\centering
\includegraphics[width=0.78\linewidth]{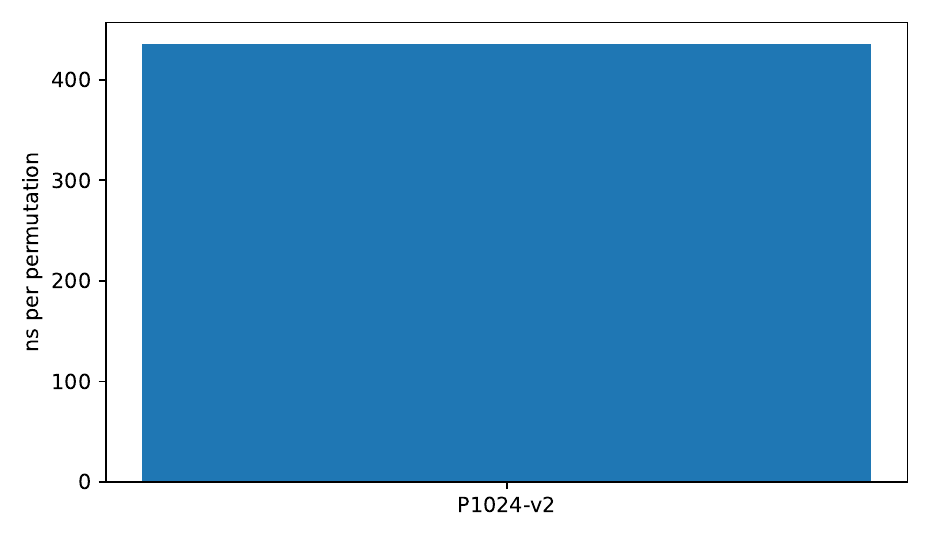}
\caption{Permutation benchmark: 435.1 ns per \textsf{P1024-v2} call (200,000 iterations).}
\label{fig:benchperm}
\end{figure}

\begin{figure}[t]
\centering
\includegraphics[width=0.78\linewidth]{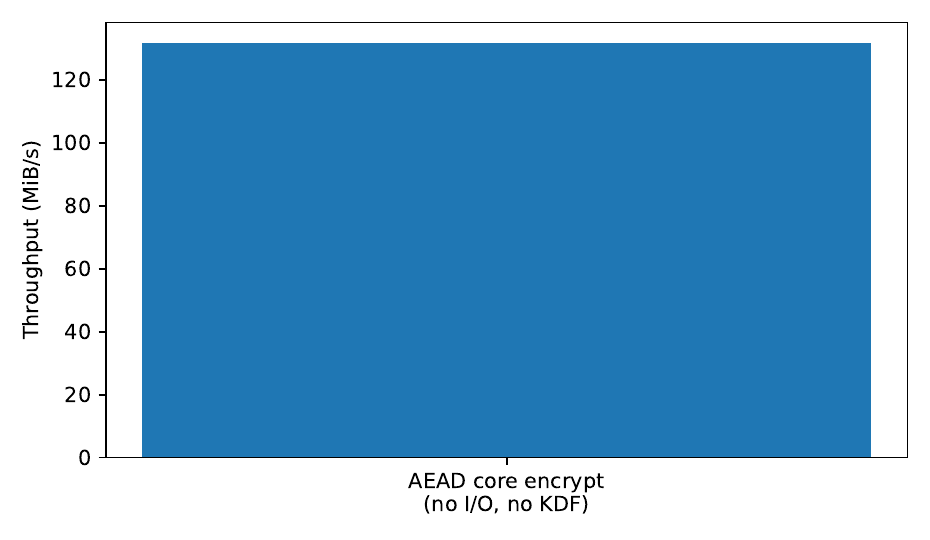}
\caption{AEAD core benchmark: 131.7 MiB/s for a 64 MiB buffer (excluding I/O and KDF).}
\label{fig:benchaead}
\end{figure}

\subsection{Avalanche and diffusion sanity check}

A diffusion sanity check measures how quickly a single-bit input difference spreads across the 1024-bit state under repeated rounds. For each trial, a random 1024-bit state $S$ is sampled, one random bit is flipped to form $S'$, and both states are iterated for $r$ rounds. The Hamming distance is recorded after each round count.

Figure \ref{fig:avalanche} shows the mean, minimum, and maximum Hamming distance across 400 trials. In this experiment, the mean approaches the random-like expectation of 512 bits after about 4 rounds, and at 24 rounds the mean is approximately 512 with a standard deviation on the order of 16 bits.

\begin{figure}[t]
\centering
\includegraphics[width=0.92\linewidth]{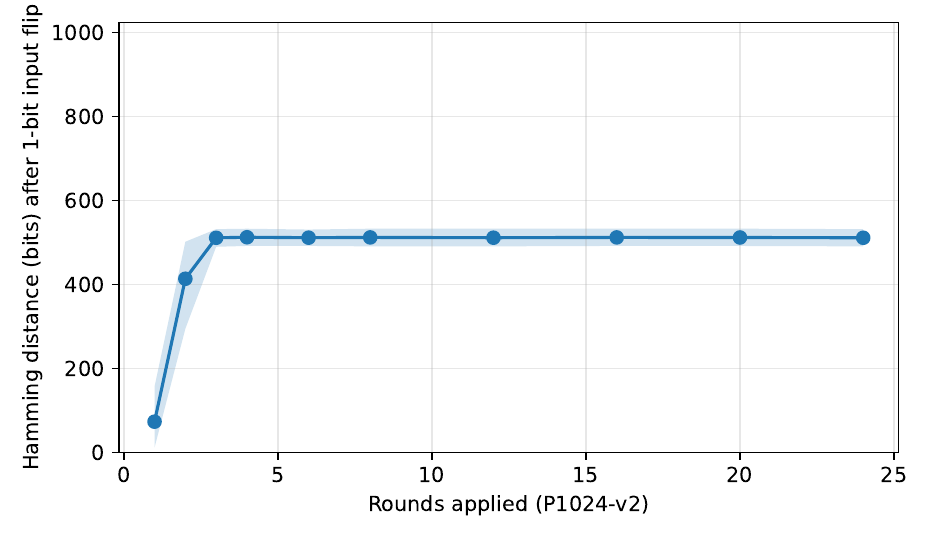}
\caption{Avalanche measurement for \textsf{P1024-v2}. This is an empirical sanity check, not a proof of security.}
\label{fig:avalanche}
\end{figure}

\subsection{Test vectors}

The reference implementation includes a \texttt{--test-all} mode that exercises encryption and decryption for multiple plaintext lengths. The produced files are included in the companion archive. Table \ref{tab:testvectors} lists SHA-256 checksums for plaintext and ciphertext files. For all tested lengths, \texttt{dec\_LEN.bin} matches \texttt{in\_LEN.bin} byte-for-byte, and ciphertext file size equals $\ell + 184$ bytes (152-byte header plus 32-byte final tag).

\input{test_vectors_table.tex}

\begin{figure}[t]
\centering
\includegraphics[width=0.92\linewidth]{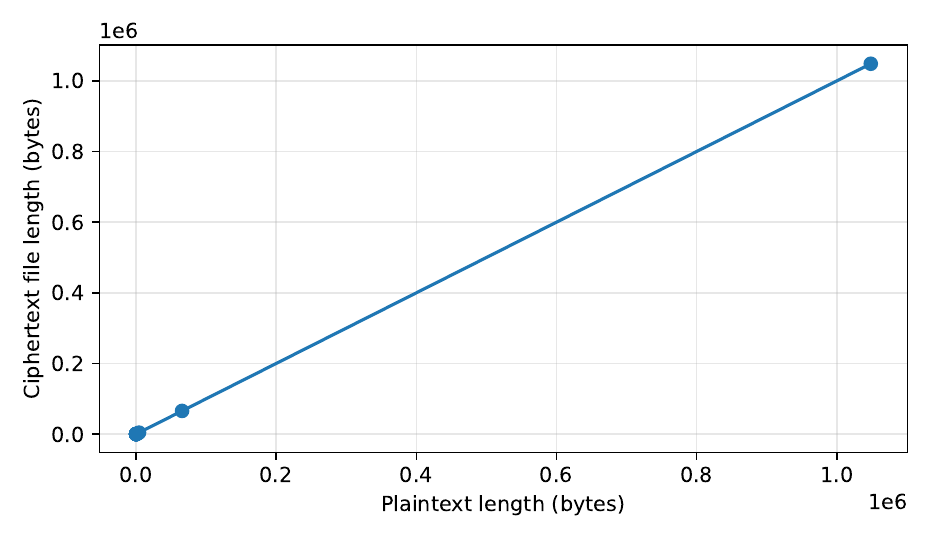}
\caption{Ciphertext file size versus plaintext length for test vectors. The overhead is constant: 152 bytes of header plus 32 bytes of final tag.}
\end{figure}

\section{Implementation Notes}

The reference code emphasizes robust engineering:
\begin{itemize}
\item Streaming encryption and decryption avoid loading whole files into memory.
\item Sensitive buffers are zeroed after use. Memory locking is attempted where available.
\item Output is written to a temporary file, \texttt{fsync} is used for durability, and an atomic rename reduces partial-output risk.
\item Tag comparisons are constant-time.
\end{itemize}

These details are essential because implementation errors are a common cause of practical cryptographic failures.

\section{Conclusion}

SymFrog-512 is a high-capacity duplex sponge AEAD built from a 1024-bit permutation and conservative parameters. This paper provides a specification aligned with the public reference implementation, a security analysis in the ideal permutation model, and empirical checks with regression artifacts and checksums. Further work includes public cryptanalysis of \textsf{P1024-v2}, tighter bounds specialized to the exact transcript structure used here, and performance evaluation across architectures.

\FloatBarrier
\appendix

\section{Reference Pseudocode}

This appendix provides compact pseudocode for reimplementation.

\subsection{\textsf{P1024-v2} round pseudocode}

\begin{lstlisting}[basicstyle=\ttfamily\small, frame=single]
Inputs: state S[0..15] of 64-bit words, round constant RC[r][0..7]

AddRoundConstants:
  for j = 0..7:
    S[8+j] ^= RC[r][j]

Mixer:
  for i = 0..7:
    S[i] ^= S[i+8]

Chi (per group of 4 words):
  for g = 0..3:
    a = S[4g+0]; b = S[4g+1]; c = S[4g+2]; d = S[4g+3]
    S[4g+0] ^= (~b) & c
    S[4g+1] ^= (~c) & d
    S[4g+2] ^= (~d) & a
    S[4g+3] ^= (~a) & b

Kick:
  for i in {0,2,...,14}:
    m = S[i] | 1
    S[i+1] ^= S[i] * m
  for i in {1,3,...,15}:
    m = S[i] | 1
    k = S[i] * (m ^ 0x9E3779B97F4A7C15)
    S[(i+1) mod 16] ^= ROTL(k,23)

Rotate and shuffle:
  for i = 0..15:
    S[i] = ROTL(S[i], (i even ? 19 : 61))
  S = (S[0], S[13], S[10], S[7], S[4], S[1], S[14], S[11],
       S[8], S[5],  S[2],  S[15], S[12], S[9], S[6],  S[3])
\end{lstlisting}

\subsection{AEAD encryption pseudocode}

\begin{lstlisting}[basicstyle=\ttfamily\small, frame=single]
Init(K,N):
  S[0..15] = LE64 words of K
  S[12..15] ^= LE64 words of N
  S[8..11] ^= ("SYMFROG-512-AEAD-v1", version word)
  S = P(S)

Absorb(DS, X):
  absorb X in 64-byte blocks into rate with 10*1 padding
  before each permutation: S[15] ^= DS; S = P(S)

Out(S):
  T[i] = S[i] ^ S[i+8] for i=0..7
  return SplitMixFinalizer(T) serialized to 64 bytes

Encrypt(AD, P):
  Init(K,N)
  Absorb(0xA0, AD)
  for each full 64-byte block:
    Z = Out(S)
    C = P ^ Z
    rate ^= C
    S[15] ^= 0xC0; S = P(S)
  final partial block:
    C_tail = P_tail ^ Out(S)[0..len-1]
    absorb ciphertext tail into rate with 10*1 padding
    S[15] ^= 0xC0; S = P(S)
  finalize:
    S[15] ^= 0xF0; S = P(S)
    T = Out(S)[0..31]
  output (C || T)
\end{lstlisting}

\section{Header Tag Pseudocode}

\begin{lstlisting}[basicstyle=\ttfamily\small, frame=single]
HeaderTag(K,N,AD,header_without_tag):
  S = Init(K,N)
  Absorb(0xB0, "SYMFROG-HDRTAG-v1" || header_without_tag || AD)
  S[15] ^= 0xB1; S = P(S)
  return Out(S)[0..31]
\end{lstlisting}

\section{Companion Archive Contents}

The companion archive includes:
\begin{itemize}
\item \texttt{main.tex} and figures under \texttt{fig/}
\item test vectors under \texttt{test\_vectors/}
\item the current \texttt{symfrog512.cpp} and \texttt{README.md} under \texttt{code/}
\end{itemize}

\section{Command Line Interface Reference}

The reference executable includes a small CLI intended for file encryption, decryption, hashing, benchmarking, and test-vector generation. The following block is taken from the built-in help text of the current reference implementation.

\begin{lstlisting}[basicstyle=\ttfamily\small, frame=single]
  symfrog512 --help
  symfrog512 --test-all
  symfrog512 --benchmark

Encrypt (AEAD):
  symfrog512 enc <in> <out> [--pass <pw> | --key-hex <hex1024>] [--ad <hex>] [--nonce-hex <hex256>] [--paranoid] [--quiet|-q]

Decrypt (AEAD):
  symfrog512 dec <in> <out> [--pass <pw> | --key-hex <hex1024>] [--ad <hex>] [--paranoid] [--quiet|-q]

Hash (FrogHash-512):
  symfrog512 hash <in> [--out <file>] [--quiet|-q]

Notes:
  --paranoid uses Argon2id SENSITIVE limits (slow, huge memory). Default is MODERATE.
  --quiet (or -q) suppresses non-error output.
  --ad is Additional Authenticated Data in hex (binds header + ciphertext).
  --nonce-hex is optional; if omitted, a random 256-bit nonce is generated.

Examples:
  symfrog512 enc secret.txt secret.syf --pass 'mypw' --ad 486561646572
  symfrog512 dec secret.syf secret.txt --pass 'mypw' --ad 486561646572
  symfrog512 hash secret.txt
\end{lstlisting}

\section{FrogHash-512 Mode}

In addition to the AEAD mode, the implementation exposes a simple sponge-based hash interface called \emph{FrogHash-512}. It uses the same 1024-bit permutation and the same rate/capacity split ($r=c=512$). At a high level:

\begin{itemize}
\item Initialize the state to all zeros, XOR a domain string (e.g., \texttt{SYMFROG-HASH-v1}) into the capacity, then apply one permutation.
\item Absorb the input file bytes in 64-byte blocks into the rate using the same $10^*1$ padding.
\item Squeeze 64 bytes from the output transform, optionally applying additional permutation calls if more output is requested.
\end{itemize}

This hash mode is included as a utility and as an additional test surface for the permutation. Like the AEAD analysis, any security claim for FrogHash-512 depends on the extent to which \textsf{P1024-v2} behaves like an ideal permutation.

\end{document}

%% file: test_vectors_table.tex
\begingroup\setlength{\tabcolsep}{0pt}\setlength{\LTpre}{0pt}\setlength{\LTpost}{0pt}\renewcommand{\arraystretch}{0.95}
\begin{longtable}{@{}r@{\hspace{6pt}}p{0.31\linewidth}@{\hspace{6pt}}p{0.31\linewidth}@{\hspace{6pt}}p{0.31\linewidth}@{}}
\caption{Regression artifacts produced by \texttt{--test-all} in the reference implementation (as provided in \texttt{testes.zip}). For every tested length, \texttt{dec\_LEN.bin} matches \texttt{in\_LEN.bin} byte-for-byte, and \texttt{enc\_LEN.syf} file size equals \(\texttt{LEN} + 184\) bytes (152-byte header plus 32-byte final tag).}\label{tab:testvectors}\\
\toprule
Len & SHA256(in\_LEN.bin) & SHA256(enc\_LEN.syf) & SHA256(dec\_LEN.bin) \\
\midrule
\endfirsthead
\toprule
Len & SHA256(in\_LEN.bin) & SHA256(enc\_LEN.syf) & SHA256(dec\_LEN.bin) \\
\midrule
\endhead
\midrule \multicolumn{4}{r}{\small Continued on next page} \\
\endfoot
\bottomrule
\endlastfoot
0 & \texttt{\tiny e3b0c44298fc1c149afbf4c8996fb924\allowbreak{}27ae41e4649b934ca495991b7852b855} & \texttt{\tiny 99c3cc8b8bec656959eef728414e6bce\allowbreak{}f0c16ad82c28c3ea623e7eace8fd519d} & \texttt{\tiny e3b0c44298fc1c149afbf4c8996fb924\allowbreak{}27ae41e4649b934ca495991b7852b855} \\
1 & \texttt{\tiny 8e35c2cd3bf6641bdb0e2050b76932cb\allowbreak{}b2e6034a0ddacc1d9bea82a6ba57f7cf} & \texttt{\tiny fd65c7c86c45e42c70e9c6a77d78dc35\allowbreak{}ea8c1cd2e5c903e8f7ea431e8793a6b0} & \texttt{\tiny 8e35c2cd3bf6641bdb0e2050b76932cb\allowbreak{}b2e6034a0ddacc1d9bea82a6ba57f7cf} \\
2 & \texttt{\tiny b07100151500cea2e51cd964a9d1ad38\allowbreak{}474d4fc37397302cfc2d9503b78fc627} & \texttt{\tiny 0487fa16898f458335f20f4256de58ec\allowbreak{}2702ecc817bdd563504db1dcceccf083} & \texttt{\tiny b07100151500cea2e51cd964a9d1ad38\allowbreak{}474d4fc37397302cfc2d9503b78fc627} \\
7 & \texttt{\tiny 494f2d94cca11a24dd3a99ecfd16042c\allowbreak{}4e0b73e68ac2eb442b0c309bbd8c4317} & \texttt{\tiny 1e315985ba88943f90d3afd6cabc11b8\allowbreak{}f5d0fb61ab770bbbf79e66a005952f44} & \texttt{\tiny 494f2d94cca11a24dd3a99ecfd16042c\allowbreak{}4e0b73e68ac2eb442b0c309bbd8c4317} \\
8 & \texttt{\tiny 65ce71f38d258aa63cb5d94564b4d902\allowbreak{}4f78a9f55a281cc10e4ebbe25fe4b26b} & \texttt{\tiny f7f5f72763fb0be206a51789252d7465\allowbreak{}b1b0b6ecdcb5d3cf9427131199014ff6} & \texttt{\tiny 65ce71f38d258aa63cb5d94564b4d902\allowbreak{}4f78a9f55a281cc10e4ebbe25fe4b26b} \\
15 & \texttt{\tiny 86f16743752818f6b54367d16142c9a0\allowbreak{}e01c22df6a01eab9942e49005467cab7} & \texttt{\tiny 8d213c12a256939f052ecc56f2d6300b\allowbreak{}b7cbb46ac514fff8b157c6fe1fbf47bd} & \texttt{\tiny 86f16743752818f6b54367d16142c9a0\allowbreak{}e01c22df6a01eab9942e49005467cab7} \\
16 & \texttt{\tiny 4ced425398e089de3d0339acc09d9bfa\allowbreak{}5e4352bf9bd6e30d29ad4777131fb4ea} & \texttt{\tiny e96fdbc46bf952462c54fbfb7f21c1a1\allowbreak{}b30290e43e7a1ed66e0d32039f9e46bc} & \texttt{\tiny 4ced425398e089de3d0339acc09d9bfa\allowbreak{}5e4352bf9bd6e30d29ad4777131fb4ea} \\
63 & \texttt{\tiny ee8654727d949ceabd1cc8c6ab03131f\allowbreak{}c14ca24319e9b69ff9aff3e8d0380660} & \texttt{\tiny dfb34e0f7316097ab96791ff5fd0137e\allowbreak{}bb3d2015b45db4133c2525ee6a8db53a} & \texttt{\tiny ee8654727d949ceabd1cc8c6ab03131f\allowbreak{}c14ca24319e9b69ff9aff3e8d0380660} \\
64 & \texttt{\tiny baf001defbd150024085a932380b34b4\allowbreak{}7c19d0ad6cba6342c87d188cbb46e050} & \texttt{\tiny 13e0ba4deb0fcfe0808edf7a1ac7d5ee\allowbreak{}4aedc49bd3b72bcad5512d655bb00af1} & \texttt{\tiny baf001defbd150024085a932380b34b4\allowbreak{}7c19d0ad6cba6342c87d188cbb46e050} \\
65 & \texttt{\tiny 11bb8329dee97c95e0046c668327d0c8\allowbreak{}821c8ee43844f2a59af0c999a6a80d49} & \texttt{\tiny 9ad4b766aaf13ecba6351a37ea21e0f8\allowbreak{}e003bcfd26cbac1a55f142ef875f4d41} & \texttt{\tiny 11bb8329dee97c95e0046c668327d0c8\allowbreak{}821c8ee43844f2a59af0c999a6a80d49} \\
127 & \texttt{\tiny a413b2926db962233d4e8065518821ad\allowbreak{}da013ea164ab02fc6a24d61faf94974f} & \texttt{\tiny 7df8a88496061348cdaa8aa1ee800ee9\allowbreak{}cf8a6a46193c65fdb743eac86e327afe} & \texttt{\tiny a413b2926db962233d4e8065518821ad\allowbreak{}da013ea164ab02fc6a24d61faf94974f} \\
128 & \texttt{\tiny da55b71156f323d3ea4bc9165866a784\allowbreak{}b6d6143718c96073ccf14bb2ffc75022} & \texttt{\tiny 37644572eff71c67bcd3e7cc3bb01150\allowbreak{}a3de1ce00765651ae86382c8f6e2a882} & \texttt{\tiny da55b71156f323d3ea4bc9165866a784\allowbreak{}b6d6143718c96073ccf14bb2ffc75022} \\
129 & \texttt{\tiny dde40808d399378fe2301e99679f1278\allowbreak{}528b5a9dbf23c3fbc676fa7f6dc30e1f} & \texttt{\tiny cafcf54172eee0e8a214beb3623aca89\allowbreak{}ac1255843230107e962768770df8b05e} & \texttt{\tiny dde40808d399378fe2301e99679f1278\allowbreak{}528b5a9dbf23c3fbc676fa7f6dc30e1f} \\
4096 & \texttt{\tiny fc6679fc6c965872eced96631d47d3a1\allowbreak{}a8e16e0dce69c5ea5bad905299a903df} & \texttt{\tiny 4bef6fca34042b6e352b0783c23eee20\allowbreak{}aa946adad704c0f7297ab610dc2509da} & \texttt{\tiny fc6679fc6c965872eced96631d47d3a1\allowbreak{}a8e16e0dce69c5ea5bad905299a903df} \\
65536 & \texttt{\tiny 6127f706e915404a2bc56bd8be27c4b7\allowbreak{}d638c690163bbc3249a98925cf997f96} & \texttt{\tiny 7fc73471e7e6f1d4c4f1365c7f1cb9bc\allowbreak{}dad27655b7057bf86c56ce73c5b7150e} & \texttt{\tiny 6127f706e915404a2bc56bd8be27c4b7\allowbreak{}d638c690163bbc3249a98925cf997f96} \\
65549 & \texttt{\tiny b6300afdd090056bb6ca8399dbbfe269\allowbreak{}31033170a3555b0752de08d8697aeab6} & \texttt{\tiny a542105d50a36b20a2f5a5e76e3f2316\allowbreak{}d3572d740ac09fc397940e7fba6aff92} & \texttt{\tiny b6300afdd090056bb6ca8399dbbfe269\allowbreak{}31033170a3555b0752de08d8697aeab6} \\
1048576 & \texttt{\tiny 2a684d85095eed4c726f5efd9a244c87\allowbreak{}3b89167b0033208ac1eacd9a1a8161c2} & \texttt{\tiny e1a234ba5084fde08a120b667a4c6d0d\allowbreak{}9ea4627cb6ef1598aac159c428a698b5} & \texttt{\tiny 2a684d85095eed4c726f5efd9a244c87\allowbreak{}3b89167b0033208ac1eacd9a1a8161c2} \\
1048583 & \texttt{\tiny 3357dd63eb743d03fd25d235fe116190\allowbreak{}1cc6664f47560f5a44caef979f1b7c98} & \texttt{\tiny 7eced6646529a3dbca96715a3a228f0e\allowbreak{}a9d09ba001f042ea466edd68a87389c8} & \texttt{\tiny 3357dd63eb743d03fd25d235fe116190\allowbreak{}1cc6664f47560f5a44caef979f1b7c98} \\
\end{longtable}
\endgroup